\documentclass[a4paper,onecolumn,11pt,accepted=2023-12-04]{quantumarticle}
\pdfoutput=1
\usepackage[utf8]{inputenc}
\usepackage[english]{babel}
\usepackage[T1]{fontenc}
\usepackage{amsmath}
\usepackage{hyperref}
\usepackage[numbers,sort&compress]{natbib}

\usepackage{tikz}
\usepackage{lipsum}
\usepackage{braket}
\usepackage{mathtools}
\usepackage{amsthm,amsfonts}
\usepackage[noend]{algpseudocode}
\usepackage[ruled,linesnumbered,vlined,noend]{algorithm2e}
\usepackage{authblk}
\usepackage{hyperref,cleveref}
\usepackage{bbm}
\usepackage{xcolor}
\usepackage{amssymb}

\newcommand{\E}{\mathbb{E}}

\newcommand{\N}{\mathbb{N}}
\newcommand{\pr}{\mathbb{P}}
\newcommand{\Oh}{\mathcal{O}}
\newcommand{\eps}{\varepsilon}

\theoremstyle{plain}
\newtheorem{theo}{Theorem}[section]
\newtheorem{lem}[theo]{Lemma}

\newtheorem{coro}[theo]{Corollary}

\newtheorem{cl}[theo]{Claim}
\theoremstyle{definition}
\newtheorem{defi}[theo]{Definition}

\newtheorem{obs}[theo]{Observation}

\newcommand{\Dk}{\Delta_{k}}
\newcommand{\Dd}{\Delta_k^\downarrow}
\newcommand{\Du}{\Delta_k^\uparrow}

\DeclareMathOperator{\Tr}{Tr}
\DeclareMathOperator{\poly}{poly}

\begin{document}

\title{A (simple) classical algorithm for estimating Betti numbers}

\author[1]{Simon Apers}
\orcid{0000-0003-3823-6804}
\author[2]{Sander Gribling}
\orcid{0000-0002-6817-2971}
\author[3]{Sayantan Sen}
\orcid{0000-0001-5875-5252}
\author[1]{D\'aniel Szab\'o}
\orcid{0000-0002-6323-9452}
\affil[1]{Universit\'e Paris Cit\'e, CNRS, IRIF, Paris, France}
\affil[2]{Tilburg University, The Netherlands}
\affil[3]{National University of Singapore, Singapore}
\maketitle

\begin{abstract}
  We describe a simple algorithm for estimating the $k$-th normalized Betti number of a simplicial complex over $n$ elements using the path integral Monte Carlo method.
        For a general simplicial complex, the running time of our algorithm is $n^{\Oh\left(\frac{1}{\sqrt{\gamma}}\log\frac{1}{\varepsilon}\right)}$ with $\gamma$ measuring the spectral gap of the combinatorial Laplacian and $\varepsilon \in (0,1)$ the additive precision.
        In the case of a clique complex, the running time of our algorithm improves to $\left(n/\lambda_{\max}\right)^{\Oh\left(\frac{1}{\sqrt{\gamma}}\log\frac{1}{\varepsilon}\right)}$ with $\lambda_{\max} \geq k$, where $\lambda_{\max}$ is the maximum eigenvalue of the combinatorial Laplacian.
        Our algorithm provides a classical benchmark for a line of quantum algorithms for estimating Betti numbers.
        On clique complexes it matches their running time when, for example, $\gamma \in \Omega(1)$ and $k \in \Omega(n)$. 
\end{abstract}

\section{Introduction}
An abstract \emph{simplicial complex} $K$ is a family of sets or \emph{faces} that is downward closed under the subset relation.
    A canonical example of a simplicial complex is the \emph{clique complex} of a graph, where the sets correspond to the cliques in the graph.
    
    There is a recent surge of interest in using simplicial complexes to model higher-order relations in data sets -- a technique often referred to as ``\emph{topological data analysis}''~\cite{carlsson2009topology}. In these studies, an object of particular interest is the homology or Betti number of a simplicial complex, which characterizes the number of high-dimensional holes in the complex.
    Unfortunately, efficiently computing the Betti number seems like a challenging task.
    Indeed, it was recently shown that even multiplicatively approximating the Betti number of a simplicial complex or clique complex is hard even for quantum computers\footnote{Multiplicative approximation of Betti numbers is QMA1-hard \cite{crichigno2022clique}, where QMA1 is the $1$-sided error version of QMA, and QMA is the quantum analog of NP.}.
    
    Thus the next natural question is whether we can \emph{additively} approximate the Betti number.
    More formally, given a parameter $\eps \in (0,1)$, can we efficiently output an estimate $\tilde\nu_k$ of the $k$-th (normalized) Betti number of the complex satisfying
    \[
    \tilde\nu_k
    = \frac{\beta_k}{d_k} \pm \eps,
    \]
    where $d_k$ denotes the number of $k$-faces in the complex.
    To the best of our knowledge, Elek \cite{elek2010betti} was the first to study this question. In \cite{elek2010betti}, it was proved that if the complex has constant degree, that is, every $0$-face (vertex) of the complex is contained in a constant number of $1$-faces (edges), then there is an algorithm whose running time depends only on the parameter~$\eps$. The constant degree assumption however implies that the complex has constant dimension as well (that is, it contains no $k$-faces for $k \in \omega(1)$), and thus all Betti numbers $\beta_k$ for non-constant $k$ are zero.
    
    The problem was later reconsidered by Lloyd, Garnerone, and Zanardi \cite{lloyd2016quantum}, who proposed a \emph{quantum} algorithm for estimating the Betti number\footnote{More correctly, \cite{lloyd2016quantum} and follow-up works output an approximation $\tilde\beta_k = \beta_k \pm \eps d_k$ of the actual Betti number, rather than the normalized version that we consider. This however requires a (multiplicative) estimate of $d_k$. While such an estimate can be obtained efficiently if e.g.~the complex is ``dense'' (i.e., $d_k/\binom{n}{k+1} = 1/\poly(n)$), we avoid this restriction.}.
    Their algorithm combines quantum techniques such as Hamiltonian simulation and quantum phase estimation.
    Assuming that we can efficiently sample a uniformly random $k$-face from the complex, the algorithm outputs an $\eps$-additive approximation of $\beta_k$ in time
    \[
    \poly(n,1/\gamma,1/\eps),
    \]
    where $n=d_0$ denotes the number of $0$-faces in the complex, and $\gamma$ is the (normalized) spectral gap of the $k$\nobreakdash-th combinatorial Laplacian.
    As current classical algorithms for calculating Betti numbers seem to run in time $\poly(d_k)$ \cite{friedman1998computing}, which can be $\poly(n^k)$, this suggests an exponential speedup of quantum algorithms over classical ones for this problem.
    This explains the surge of interest in the quantum algorithm, and in particular in its application to clique complexes, which have a concise $\poly(n)$\nobreakdash-size description \cite{akhalwaya2022exponential,crichigno2022clique,mcardle2022streamlined,berry2022quantifying,schmidhuber2022complexity,ameneyro2022quantum,hayakawa2021quantum}.

    \paragraph{Contributions:} 
    In this work, we describe a simple classical algorithm for approximating Betti numbers. We refer to \Cref{thm:main} for a formal statement about the complexity of our algorithm, here we discuss its implications. Our algorithm provides a natural benchmark for the aforementioned quantum algorithms.
Similarly to these quantum algorithms, our algorithm is efficient if the gap and the precision are ``\emph{large}''.
However, while the quantum algorithms can afford a gap $\gamma$ and precision $\eps$ as small as $1/\poly(n)$, our algorithm requires these to be \emph{constant} for general complexes.
This is similar to the dequantization results from \cite{gharibian2022dequantizing}.
In the case of clique complexes, we can go further.
For example, if $k \in \Omega(n)$ then we can afford precision $\eps = 1/\poly(n)$ (if the gap is constant) or gap $\gamma = \Omega(1/\log^2 n)$ (if the precision is constant).
Overall, this further pins down the region where we can expect a quantum advantage for the problem of estimating Betti numbers.

In a nutshell, we base our algorithm on a random variable whose expectation is close to $\beta_k/d_k$ and whose variance is (sufficiently) small. Crucially, we show that we can efficiently generate samples from this random variable. Standard concentration bounds can be used to bound the required number of samples and hence establish the complexity of our algorithm.
    More precisely, the algorithm is based on the technique of path integral Monte Carlo~\cite{barker79}, akin to the Ulam-von Neumann algorithm for solving linear systems \cite{forsythe1950matrix}.
    Our result is formally stated below.
    By $\Gamma(\Dk)$ we denote the \emph{spectral gap} of the combinatorial Laplacian $\Dk$, which is equal to its smallest nonzero eigenvalue.
    
    \begin{theo} \label{thm:main}
        Let $\Dk$ denote the $k$-th combinatorial Laplacian of the complex.
        Assume that in time $\poly(n)$ we can {\normalfont (i)} draw a $k$-face uniformly at random, and {\normalfont (ii)} check whether a set is in the complex.
        Given an estimate $\hat\lambda \in \Theta(\lambda_{\max}(\Dk))$ and a lower bound $\gamma$ such that $\Dk$ has spectral gap $\Gamma(\Dk) \geq \gamma \hat\lambda$, there exists a classical algorithm that, for any $\eps > 0$, outputs with high probability an estimate $\tilde\nu_k = \beta_k/d_k \pm \eps$ of the $k$-th (normalized) Betti number of a general simplicial complex in time
        \[
        n^{\Oh\left( \frac{1}{\sqrt{\gamma}} \log \frac{1}{\eps} \right)},
        \]
        and of a clique complex in time
        \[
        \left(\frac{n}{\hat\lambda}\right)^{\Oh\left( \frac{1}{\sqrt{\gamma}} \log \frac{1}{\eps} \right)}
        \cdot \poly(n).
        \]
        The algorithm has space complexity $\poly(n,1/\gamma,\log(1/\eps))$.
    \end{theo}

For general complexes, our algorithm improves upon the aforementioned classical algorithms if $k\in\Omega(1/\sqrt{\gamma})$. Since $n \geq \lambda_{\max}(\Dk) \geq k + \delta_k+1$ (see \Cref{lem:D-facts}), with $\delta_k$ being the maximum up-degree\footnote{The up-degree of a $k$-face is the number of $(k+1)$-faces that contain the $k$-face (see \Cref{def:up_degree}).} over all $k$-faces, we can simply set $\hat\lambda = n$ if $k \in \Omega(n)$ or if we know that $\delta_k \in \Omega(n)$.
In such case, the algorithm for clique complexes runs in time $2^{\Oh\left( \frac{1}{\sqrt{\gamma}} \log \frac{1}{\eps} \right)} \cdot \poly(n)$.
This is polynomial if either $\gamma \in \Omega(1)$ and $\eps = 1/\poly(n)$, or $\gamma \in \Omega(1/\log^2 n)$ and $\eps \in \Omega(1)$.
The algorithm provides a classical counterpart to the aforementioned line of quantum algorithms for estimating Betti numbers which, under similar assumptions, have a runtime scaling as $\poly(n,1/\gamma,1/\eps)$.
We summarize these findings in \Cref{tab:complexities}.

\begin{table}[]
    \centering
    \begin{tabular}{|c|c|c|}\hline
         \textbf{Algorithm} & \textbf{Complexes} & \textbf{Complexity is $\poly(n)$ if}  \\ \hline
         quantum algorithm of \cite{lloyd2016quantum} & general & $\gamma,\eps\in \Omega(1/\poly(n))$ \\ \hline
         this work & general & $\gamma, \eps \in \Omega(1)$  \\ \hline
         this work & clique, $k \in \Omega(n)$
         & \begin{tabular}{c} $\gamma\in \Omega(1), \eps\in \Omega(1/\poly(n))$ \\ or $\gamma\in \Omega(1/\log^2(n)), \eps \in \Omega(1)$ \end{tabular}
         \\ \hline
    \end{tabular}
    \caption{Comparison of the parameter settings of quantum and classical algorithms for the Betti number estimation problem under which their running time is polynomial.}
    \label{tab:complexities}
\end{table}

 The complexity of our algorithm for general simplicial complexes can alternatively be obtained using the singular value transformation (SVT) algorithm\footnote{More precisely, the ``\emph{dequantized quantum singular value transformation algorithm}'' in Section $4$ of \cite{gharibian2022dequantizing}.} from Gharibian and Le Gall~\cite{gharibian2022dequantizing}.
The main difference is that we use a path integral Monte Carlo approach for computing matrix powers, instead of computing them explicitly as in~\cite[Lemma 3]{gharibian2022dequantizing}.
This approach provides us with an exponential improvement in the space complexity since the SVT algorithm has space complexity $n^{\Oh\left( \frac{1}{\gamma} \log \frac{1}{\eps} \right)}$.
The more significant benefit is that we get an improved algorithm for clique complexes, which is the main case of interest for the aforementioned quantum algorithms.
We show that the $k$-th combinatorial Laplacian is $n$-sparse for clique complexes, as compared to general complexes which are $\Oh(kn)$-sparse.
This implies that it is closer to a diagonally dominant matrix, and we can exploit this for designing better algorithms using the path integral Monte Carlo technique.

\paragraph{Open problems}

A natural question is to what extent we can improve our results.
The most stringent barrier seems to come from~\cite[Theorem 6]{cade2021complexity}, who proved that Betti number estimation  for general complexes is DQC1-hard when $\eps,\gamma = 1/\poly(n)$, where DQC1 is a complexity class that is expected to be hard to simulate classically\footnote{In fact, they consider a slight generalization of the problem called ``quasi-Betti number estimation''.}.
This safeguards a quantum speedup for the case of general complexes, yet it leaves open the case of clique complexes.
Our work shows that we can get additional leverage for clique complexes.
We leave it as our main open question whether the classical complexity for clique complexes can be improved to $\poly(n,1/\gamma,1/\eps)$.

The task of estimating persistent Betti numbers has been getting an increasing amount of attention, see e.g.~\cite{persistent_spectrum, persistent_laplacian, hayakawa2021quantum}. It seems like our method can be used for solving this problem too (if we have membership query access to both complexes $K$ and $L$ with $K\subseteq L$), but we leave the formal proof of this for future work. However, note that we lose the speedup in the clique complex case, because for the persistent Laplacian it is no longer true that it can only have less than $n$ nonzero entries per row (see our \Cref{lem_clique_nonzeros}). For example, if $K$ only consists of $n$ isolated vertices, but $L$ is the clique complex corresponding to the complete graph, then for any $k$ the $k$-th persistent Laplacian of the pair $(K,L)$ has $\Omega(nk)$ nonzero entries in each row.

A final open question, as was already mentioned in earlier works \cite{berry2022quantifying}, is characterizing which complexes admit a large spectral gap.
The advantage of our algorithm, as well as the aforementioned quantum algorithms, hinges on this assumption.
We note that \cite{berry2022quantifying} discussed the complete $k$-partite graph $K(m,k)$ as an example where our spectral gap assumption on the combinatorial Laplacian holds. 
$K(m,k)$ consists of $k$ clusters, where each cluster contains $m$ vertices, and two vertices are adjacent iff they are in different clusters. The $(k-1)$-th combinatorial Laplacian of the clique complex defined by $K(m,k)$ has spectral gap $m$ and the $(k-1)$-th Betti number of the complex is $(m-1)^k$ (see \cite[Proposition 1 \& 2]{berry2022quantifying}). Observe that for $K(n/k,k)$, the spectral gap is $n/k$, and the normalized spectral gap is $\gamma=1/k$.
Thus, for this kind of complexes our algorithm runs in polynomial time if for example $\eps\in\Omega(1)$ and $k\in O\left(\log^2(n)\right)$.
    
    \paragraph*{Related work.}
    
    Upon completion of this work, we noticed that a similar classical path integral Monte Carlo algorithm for estimating Betti numbers was proposed in \cite{berry2022quantifying}.
    The authors use a Trotterization approach to implement an imaginary time evolution of the combinatorial Laplacian, and use a more complex distribution over paths to minimize the variance of the Monte Carlo estimator.
    These techniques seem more flexible than ours, and might eventually lead to a better algorithm.
    However, unless some additional conditions are imposed, the current runtime of their algorithm still shows an exponential dependency on both $k$ and $1/\eps$, which our algorithm avoids.

    As a reviewer has pointed out, a similar usage of the path integral Monte Carlo method to estimate elements of a matrix power is present in \cite{doron2017approximating} (in particular, see their Lemma 2.5).

\section{\label{sec:prelim} Notations and Preliminaries}

\subsection{Homology and Laplacians \label{sec:homology}}

In this work, we consider a simplicial complex $K$ over the set $[n]$, which corresponds to a downward closed set system over $[n]$, where $[n]=\{1, \dots, n\}$. We denote by $C_k \subseteq K$ the set of subsets $A = \{i_1,\dots,i_{k+1}\} \subseteq [n]$ of cardinality $k+1$, also called the \emph{$k$-faces} of $K$, and $d_k=|C_k|$.
The \emph{dimension} of $K$ is the largest integer $k$ such that $K$ contains at least one $k$-face, that is, $\dim(K)=\max\{k\text{ s.t. }d_k > 0\}$.
We \emph{orient} the faces by ordering the vertices in increasing order (that is, we assume $i_\ell < i_{\ell+1}$).
Once we have defined an orientation of the faces, we associate to each face a basis vector $\ket{A} \coloneqq \ket{i_1\dots i_{k+1}}$.
The \emph{boundary operator} $\partial_k:C_k \mapsto C_{k-1}$ is defined by
\[
\partial_k \ket{A}
= \sum_{\ell=1}^{k+1} (-1)^{\ell-1} \ket{A \setminus \{i_\ell\}}.
\]
From the boundary operator we can define the \emph{combinatorial Laplacian} $\Dk:C_k \mapsto C_k$ as 
\begin{gather*}
    \Dk
= \Dd + \Du,
\\ 
\text{ with } \quad
\Dd = \partial_k^\dag \partial_k
\quad \text{and} \quad
\Du = \partial_{k+1} \partial_{k+1}^\dag.
\end{gather*}

$\Dk$ is a symmetric, positive semidefinite matrix of size $d_k\times d_k$. Next we define notions of degree and neighborhood of a face.

\begin{defi} \label{def:up_degree}
    In a simplicial complex, the \emph{up-degree} of a $k$-face $i$ is the number of $(k+1)$-faces that contain $i$. It is denoted as $d^{\text{up}}_i \coloneqq |\{j\in C_{k+1} \text{ s.t. } i\subseteq j\}|$.
    The maximum up-degree among all the $k$-faces is denoted as $\delta_k = \max_{i\in C_k} d^{\text{up}}_i$.
\end{defi}

\begin{defi}[Down-up and up-down neighbors]
    Let $i, j\in C_k$ be two $k$-faces of a simplicial complex $K$. $i$ and $j$ are said to be \emph{down-up neighbors} if their symmetric difference $|i \triangle j|=2$.
    Additionally, if $i \cup j$ is a $(k+1)$-face of $K$, then $i$ and $j$ are also said to be \emph{up-down neighbors}.
\end{defi}

The following lemma from \cite{goldberg2002combinatorial} uses these notions to characterize the entries of $\Dk$.

\begin{lem}[Restatement of Laplacian Matrix Theorem, {\cite[Theorem 3.3.4]{goldberg2002combinatorial}}]\label{lem:laplacemat}
    Let $K$ be a finite oriented simplicial complex, $k$ be an integer with $0 < k \leq \mathrm{dim}(K)$, and $\{1, 2, \ldots, d_k\}$
denote the $k$-faces of $K$. Let $i, j \in [d_k]$. Then we have the following:

\begin{itemize}
    \item $(\Dk)_{ii}= d^{\text{up}}_i + k +1$.

     \item $(\Dk)_{ij} = \pm 1$ if  $i \neq j$,  and $i$ and $j$ are down-up neighbors but they are not up-down neighbors.

    \item $(\Dk)_{ij} = 0$ otherwise (i.e. if $i \neq j$, and either $i$ and $j$ are up-down neighbors or they are not down-up neighbors).
\end{itemize}
\end{lem}

\noindent
The following lemma gathers some useful facts about $\Dk$ that will be used in our proofs.

\begin{lem} \label{lem:D-facts}
    Let us consider a simplicial complex $K$ with $\Delta_k$ being its combinatorial Laplacian and $\delta_k$ being the maximum up-degree among all $k$-faces of $K$. Then the following results hold:   
    \begin{itemize}
        \item
        $\delta_k + k + 1
        \leq \lambda_{\max}(\Delta_k)
        \leq n
        $.
        
        \item$
        (\Dk)_{ii}
        \leq n$.
        
        \item
        $\Dk$ has at most $(n-k-1)(k+1)$ nonzero off-diagonal entries in each row, all equal to $\pm 1$ \footnote{This is tight up to a constant for general simplicial complexes, in contrast to some earlier papers \cite{gyurik2022towardsquantum,mcardle2022streamlined} that mention $\Oh(n)$ nonzero off-diagonal entries. Later in the paper (see \Cref{lem_clique_nonzeros}) we show  that for clique complexes it is actually~$\Oh(n)$, and we exploit this to design a more efficient algorithm.}.
    \end{itemize}
\end{lem}
\begin{proof}
    The second and third bullets follow from \Cref{lem:laplacemat}. The second inequality of the first bullet follows from \cite[Proposition 6.2]{duval2002shifted}, who prove that $\lambda_{\max}(\Dk^\uparrow),\lambda_{\max}(\Dk^\downarrow) \leq n$.
    This gives the claimed bound if we use that $\lambda_{\max}(\Dk) = \max\{\lambda_{\max}(\Dk^\uparrow),\lambda_{\max}(\Dk^\downarrow)\}$, which follows from $\Dk^\uparrow \Dk^\downarrow = \Dk^\downarrow \Dk^\uparrow = 0$.

    For proving the first inequality of the first bullet, we write the largest eigenvalue using the Rayleigh quotient:
    \begin{align*}
        \lambda_{\max}(\Delta_k)
        &= \max_{\|x\|_2=1}x^T\Dk x =\max_{\|x\|_2=1}(x^T\Du x+ x^T\Dd x) \\
        &= \max_{\|x\|_2=1}(\|\partial_{k+1}^\dag x\|_2^2+ \|\partial_k x\|_2^2).
    \end{align*}
    
    We can lower bound this by taking a particular $x$: the one that is all zero except for a position where it is one, and the latter position corresponds to a $k$-face with up-degree $\delta_k$. For this vector, $\partial_{k+1}^\dag x$ contains $\delta_k$ ones and the other elements are zero (because it has up-degree $\delta_k$). And it is also true that $\partial_k x$ contains $k+1$ ones and the other elements are zero (because every $k$-face contains $k+1$ many $(k-1)$-faces). This concludes that $\delta_k + k + 1 \leq \lambda_{\max}(\Delta_k)$.
\end{proof}

\noindent
Now let us define the central objects of our interest, namely the \emph{Betti numbers} of a simplicial complex.

\begin{defi}[Betti number]
    Let $\Delta_k$ be the $k$-th combinatorial Laplacian of a simplicial complex~$K$ with $\partial_k$ denoting the boundary operator.
    The $k$-th Betti number $\beta_k$ equals the dimension of the $k$-th homology group: $\beta_k=\dim(\ker(\partial_k) / \mathrm{Im}(\partial_{k+1}))$.
    Equivalently, $\beta_k = \dim \ker(\Dk)$, thus it is equal to the number of zero eigenvalues of $\Dk$.
\end{defi}

We will be particularly interested in the Betti numbers of \emph{clique complexes}.
The clique complex~$K$ associated to a graph $G = (V,E)$, with vertices $V = [n]$ and (undirected) edges $E$, is the family of subsets of $V$ that are cliques in $G$, that is, a subset $A\subseteq V$ is in $K$ if and only if $A$ induces a clique in $G$.

\subsection{Hoeffding's inequality}
The following concentration bound will be crucial in our proofs.

\begin{lem}[Hoeffding's Inequality~\cite{hoeffding63}, see also~\cite{dubhashi2009concentration}] \label{lem:hoeffdingineq}
    
    Let $X_1,\ldots,X_p$ be independent random variables such that $a \leq X_i \leq b$ and let $X= \frac{1}{p} \sum\limits_{i=1}^p X_i$. Then, for all $\delta >0$, 
    $$ \pr\left(|X-\E[X]| \geq \delta\right) \leq  2\exp\left(\frac{-2 p \delta^2} {(b-a)^2}\right).
    $$
\end{lem}

\section{Algorithm for general simplicial complexes} \label{sec:general_algo}

Consider a general simplicial complex $K$ over $[n]$, with $d_k$ being the number of its $k$-faces, $\Dk$ its $k$-th combinatorial Laplacian, and with $k$-th Betti number $\beta_k$.
We wish to obtain an estimate~$\tilde\nu_k$ that satisfies $\tilde\nu_k = \beta_k/d_k \pm \eps$ for some parameter $\eps \in (0,1)$. In this section, we make the following assumptions: 
\begin{enumerate}
    \item
    In time polynomial in $n$, (a) we can check whether a set is in the complex, and (b) we can sample a $k$-face from the simplicial complex $K$ uniformly at random \footnote{We do not need assumption (b) if the complex is dense in $k$-faces.}.
    \item
    We have estimates $\hat\lambda$ and $\gamma$ on the largest eigenvalue and the spectral gap of $\Dk$, respectively, satisfying
    \[
    \lambda_{\max}(\Dk)
    \leq \hat\lambda \leq c \lambda_{\max}(\Dk)
    \quad \text{ and } \quad
    \lambda_2(\Dk) \geq \gamma \hat\lambda,
    \]
    for some constant $c > 0$ \footnote{If we do not have such a bound on the spectral gap, an alternative is to approximate the ``quasi-Betti number'', which is the number of small eigenvalues (below $\gamma \hat\lambda$) of the combinatorial Laplacian, as in \cite{cade2021complexity}.}.
\end{enumerate}

Now consider the matrix
\[
H
= I - \Dk/\hat\lambda,
\]
which, as we discuss below, satisfies $0 \preceq H \preceq I$.
From \Cref{lem:D-facts},  we know that $0\leq (\Dk)_{ii} \leq n$ and $\Dk$ has $\Oh(nk)$ nonzero off-diagonal entries in every row, each of absolute magnitude 1.
This implies that
\[
\| H \|_1
= \max_j \sum_i |H_{ij}|
\in \Oh(n k).
\]
By construction, the $k$-th combinatorial Laplacian $\Dk$ is positive semidefinite, hence all eigenvalues are non-negative, $\beta_k$ of them are equal to 0, the second smallest distinct eigenvalue is $\lambda_2(\Dk)$, and the maximum eigenvalue is $\lambda_{\max}(\Dk)$. Thus, by linearity, the eigenvalues of $H$ lie between $0$ and~$1$, $\beta_k$ of them equal to $1$, and all other eigenvalues lie below $1 - \lambda_2(\Dk)/\hat\lambda \leq 1 - \gamma$.
The following lemma shows how to relate the trace of $H^r$, for sufficiently large $r$, to the Betti number $\beta_k$.

\begin{lem} \label{obs:tr-bk}
    If $r \geq \frac{1}{\gamma} \log\frac{1}{\eps}$ then $\beta_k \leq \Tr\left( H^r \right) \leq \beta_k + \eps d_k$.
\end{lem}

\begin{proof}
    On the one hand, we have that
    \[
    \Tr\left( H^r \right)
    = \sum_{i=1}^{d_k} \lambda_i(H)^r
    \geq \beta_k.
    \]
    On the other hand, we have that
    \[
    \Tr\left( H^r \right)
    = \sum_{i=1}^{d_k} \lambda_i(H)^r
    \leq \beta_k
    + \sum_{i:\lambda_i(H) < 1} (1-\gamma)^r
    \leq \beta_k + \eps d_k,
    \]
    where we used that $(1-\gamma)^r \leq \eps$ for $r \geq \frac{1}{\gamma} \log\frac{1}{\eps}$.
\end{proof}
\noindent

Using this observation, we can obtain a $2\eps$-additive estimate of $\beta_k/d_k$ from an $\eps$-additive estimate of $\Tr\left( H^r \right)/d_k$.
To obtain the latter we use another observation, that holds not only for large $r$ as above, but for a general nonnegative $z$-th power of $H$.

\begin{obs}\label{obs:traceHr}
    \[
    \frac{1}{d_k} \Tr\left( H^z \right)
    = \frac{1}{d_k} \sum_{i=1}^{d_k} \braket{i|H^z|i}
    = \E_i\left[X_z^{(i)}\right],
    \]
    where $i \in [d_k]$ is sampled uniformly at random and $X_z^{(i)} = \braket{i|H^z|i}$.
\end{obs}

Since $H$ is $\Oh(n k) \in \Oh(n^2)$-sparse, we can evaluate $X_z^{(i)}$ exactly in time $\Oh(n^{2z})$.
Indeed this is the approach in \cite{gharibian2022dequantizing}. Here we use another approach based on the path integral Monte Carlo method, which has two advantages.
First, it improves the space complexity from $n^{\Oh(z)}$, as in \cite{gharibian2022dequantizing}, to~$\widetilde\Oh(n z)$.
Second, it will lead to a faster algorithm for clique complexes (see next section).

Let us denote the sign of $H_{i,j}$ by $(-1)^{s(i,j)}$, with $s(i,j) \in \{0,1\}$.
We can rewrite $X_z^{(i)}$ as follows:
\begin{align*}
    X_z^{(i)}
    &= \braket{i|H^z|i} \\
    &= \sum_{j_1,\dots,j_z} \braket{i|j_z} \braket{j_z|H|j_{z-1}}
    \dots \braket{j_1|H|i} \\
    &= \sum_{j_1,\dots,j_z} Y_z(i,j_1,\dots,j_z)
    \frac{|H_{j_z,j_{z-1}}|}{\|H_{\cdot,j_{z-1}}\|_1} \dots \frac{|H_{j_1,i}|}{\|H_{\cdot,i}\|_1},
\end{align*}
with
\[
Y_z(j_0,j_1,\dots,j_z)
= \braket{j_0|j_z} \prod_{\ell=0}^{z-1} (-1)^{s(j_{\ell+1},j_\ell)} \|H_{\cdot,j_\ell}\|_1.
\]
By \Cref{lem:D-facts} it holds that $|Y_z| \leq \|H\|_1^z \leq (n + n k)^z \in \Oh(n^{2z})$.
By our choice of normalization, we can interpret $|H_{j,i}|/\|H_{\cdot,i}\|_1 \eqqcolon P(i,j)$ as a transition probability from face $i$ to face $j$.
We can then say that
\[
X_z^{(i)}
= \E_{(j_0=i,j_1,\dots,j_z)}\left[ Y_z(j_0,j_1,\dots,j_z) \right],
\]
where the path $(j_0,j_1,\dots,j_z)$ is drawn with probability $P(j_0,j_1) \dots P(j_{z-1},j_z)$ from the resulting Markov chain with transition matrix $P$.
Moreover, if we choose the initial $k$-face $j_0 \in [d_k]$ uniformly at random, then
\begin{align*}
    \E\left[Y_z\right]
&=
\E_{(j_0,j_1,\dots,j_z)}\left[ Y_z(j_0,j_1,\dots,j_z) \right] \\
&= \E_{j_0}\left[X_z^{(j_0)}\right]
= \frac{1}{d_k} \Tr(H^z).
\end{align*}
This gives us an unbiased estimator $Y_z$ for the normalized trace of $H^z$.
Moreover, as proven in the following lemma, we can sample $Y_z$ efficiently.

\begin{lem} \label{lemma_sample}
    We can sample from $Y_z$, as defined above, in time $z \cdot \poly(n)$.
\end{lem}

\begin{proof}
    We can evaluate $Y_z$ by sampling $z$ steps of the Markov chain over $k$-faces.
    The initial $k$-face $j_0$ is drawn uniformly at random.
    By our assumptions, we can do this in time $\poly(n)$.
    Subsequent steps are sampled as follows.
    
    Let $j_i$ be the current $k$-face.
    First we learn the up-degree $d^{up}_{j_i}$, and hence $(\Dk)_{j_ij_i}$.
    We do this by, for all potential up-neighbors (obtained by adding one element to the face), querying whether they are in the complex.
    This takes $n-k-1$ queries, and hence time $\poly(n)$.
    Then we learn all down-up neighbors by querying all $\Oh(n^2)$ subsets with symmetric difference 2.
    This again takes time $\poly(n)$.
    By \Cref{lem:laplacemat} we can now derive all $\Oh(n^2)$ nonzero entries of the $j_i$-th row $H_{\cdot,j_i}$, and hence sample $j_{i+1}$ according to the probability $P(j_{i+1},j_i) = |H_{j_{i+1},j_i}|/\|H_{\cdot,j_i}\|_1$.
    This yields time $\poly(n)$ per step of the Markov chain, and so $z \cdot \poly(n)$ time overall.
\end{proof}

It remains to bound the complexity of estimating $\E[Y_z]$, given samples of $Y_z$.
For this, we use Hoeffding's inequality (\Cref{lem:hoeffdingineq}), which yields the following lemma.

\begin{lem} \label{lem:Hoeffding2}
For any $\delta > 0$ and integer $z \geq 0$, we can obtain a $\delta$-additive estimate of $\E[Y_z] = \Tr(H^z)/d_k$ by taking the average of $\Oh(n^{4z})/\delta^2$ many independent samples of $Y_z$.
\end{lem}

\begin{proof}
    
    We know that $|Y_z| \leq \|H\|_1^z \leq (n + n k)^z \in \Oh(n^{2z})$ (\Cref{lem:D-facts}). 
    Consider $p$ independent samples $Y_{z,1}, \ldots, Y_{z,p}$ distributed according to $Y_z$.
    For any $\delta > 0$, Hoeffding's inequality (\Cref{lem:hoeffdingineq}) states that 
    \[
    \Pr\left( \left| \frac{1}{p} \sum_{i=1}^p Y_{z,i} - \E[Y_z] \right|
    \geq \delta \right)
    \leq 2 \exp\left(\frac{-2 p \delta^2}{\Oh(n^{4z})}\right).
    \]
    If we choose $p = \Oh(n^{4z})/\delta^2$ then $\frac{1}{p} \sum Y_{z,i}$ will be $\delta$-close to its expectation $\E\left[ Y_z \right] = \frac{1}{d_k} \Tr(H^z)$ with probability at least $1-1/2^{\poly(n)}$. 
    \end{proof}
    
This leads to the following algorithm, which has time complexity~$\Oh(n^{4z}/\delta^2)$.

\begin{algorithm}[h]
\caption{Algorithm for $\delta$-estimating $\Tr\left(H^z\right)/d_k = \Tr\left(\left(I-\Dk/\hat\lambda\right)^z\right)/d_k$}
\label{alg:simple}
\KwIn{Query and sample access to complex $K$, integer $k$, parameters $\hat\lambda$ and $z$, precision parameter $\delta \in (0,1)$.}
\KwOut{Estimate $\text{est}_{k,z}$ such that $\text{est}_{k,z} = \Tr(H^z)/d_k \pm \delta$ with high probability.}
\SetAlgoLined
\SetKw{Continue}{continue}

\vspace{6pt}

Set $p = \Oh(n^{4z})/\delta^2$.

\For{$t = 1,\dots,p$}{
    
    Sample a $k$-face $j_0$ of $K$ uniformly at random.\
    
    Sample $z$ steps $(j_0,j_1,\dots,j_z)$ of the Markov chain $P$ with initial face $j_0$.\
        
        Set $Y_{z,t}
        = \braket{j_0|j_z} \prod_{q=0}^{z-1} (-1)^{s(j_{q+1},j_q)} \|H_{\cdot,j_q}\|_1$.\
    }
    
Return $\text{est}_{k,z} = \frac{1}{p} \sum_{t=1}^{p} Y_{z,t}$.\

\end{algorithm}

For $r \geq \frac{1}{\gamma} \log \frac{2}{\eps}$, we know from \Cref{obs:tr-bk} that $\Tr(H^r)/d_k = \beta_k/d_k \pm \eps/2$.
Hence, setting $\delta = \eps/2$ and $z=r$ in the algorithm above we get an $\eps$-additive estimate of $\beta_k/d_k$.
The algorithm requires $p = \Oh(n^{4r}/\delta^2) = n^{\Oh\left(\frac{1}{\gamma} \log\frac{1}{\epsilon}\right)}$ samples of $Y_r$, each of which can be obtained in time $r \cdot \poly(n)$ by \Cref{lemma_sample}.
The overall time complexity of Algorithm \ref{alg:simple} is hence~$n^{\Oh\left(\frac{1}{\gamma} \log\frac{1}{\eps}\right)}$.

\subsection{Improvement using Chebyshev polynomials}

We can slightly improve this result by approximating $H^r$ with a polynomial of degree roughly $\sqrt{r}$ using Chebyshev polynomials, and then estimating the monomials using Algorithm \ref{alg:simple}.
We use the following approximation lemma.

\begin{lem}[Follows from {\cite[Theorem 3.3]{Sachdeva_Chebyshev}}] \label{lem:cheb-approx}
    For any $\delta>0$ and $d\ge \sqrt{2r\log(2/\delta)}$, the monomial $x^r$ can be approximated by a polynomial $p_{r,d}(x)$ of degree $d$ such that $|p_{r,d}(x)-x^r|\le \delta$ for all $x \in [-1,1]$.
\end{lem}

Now we bound the size of the monomial coefficients in $p_{r,d}(x)$, as these will govern the precision with which we need to estimate the trace of each monomial.
Following \cite[Chapter 3]{Sachdeva_Chebyshev}, the polynomial $p_{r,d}(x)$ is obtained by first approximating $x^r$ in the Chebyshev basis by  
\[
p_{r,d}(x)
= \alpha^{(r)}_0 + \sum_{i=1}^d 2\alpha^{(r)}_i T_i(x)
\]
where $T_i(x)$ is the $i$'th Chebyshev polynomial (of the first kind), and $\alpha^{(r)}_i = \binom{r}{(r-i)/2}/2^r$ if $i$ has the same parity as $r$, and $\alpha^{(r)}_i = 0$ otherwise.
We then obtain the desired coefficients of $p_{r,d}$ in the monomial basis by expressing each of the Chebyshev polynomials in the monomial basis.
Concretely, if
$T_i(x) = \sum_{\ell=0}^i c_\ell^{(i)} x^\ell$
then
\begin{align*}
    p_{r,d}(x)
&= \alpha^{(r)}_0 + \sum_{i=1}^d 2 \alpha^{(r)}_i T_i(x) \\
&= \alpha^{(r)}_0 + \sum_{\ell=0}^d \left[\sum_{i=\ell}^{d} 2\alpha^{(r)}_i c_\ell^{(i)}\right] x^{\ell}
\eqqcolon \sum_{\ell=0}^d b^{(r,d)}_\ell x^\ell.
\end{align*}
 
To bound the coefficients $b^{(r,d)}_\ell$, we first bound the coefficients $c_\ell^{(i)}$ in \Cref{lem:c-bound} below.
Combined with the bounds $\left|\alpha^{(r)}_i\right|\leq 1$ this lemma yields the upper bound $|b_\ell| \le (d+1) \cdot 2 \cdot 2^{2d}\le 2^{3d}$.

\begin{lem} \label{lem:c-bound}
    For all $i \in \N$ and $\ell \leq i$, we have $\left|c_\ell^{(i)}\right|\le 2^{2i}$.
\end{lem}
\begin{proof}
We prove the lemma using induction on $i$. The Chebyshev polynomials can be defined via the following recurrence
$$T_i(x)=2xT_{i-1}(x)-T_{i-2}(x),$$
where $T_0(x)=1$, $T_1(x)=x$. 
    This immediately shows that the bounds $\left|c_{\ell}^{(i)}\right| \leq 2^{2i}$ hold for $i=0$ and $i=1$. Now assume $i>1$ and that for all $i'<i$ and $\ell \leq i'$, we have $\left|c_\ell^{(i')}\right|\le 2^{2i'}$. Then from the recursion $T_i(x)=2xT_{i-1}(x)-T_{i-2}(x)$, we obtain $c_\ell^{(i)}=2c_{\ell-1}^{(i-1)}-c_{\ell}^{(i-2)}$ and hence $\left|c_\ell^{(i)}\right| \le 2^{2(i-1)+1}+2^{2(i-2)}\le 2\cdot 2^{2(i-1)+1}= 2^{2i}$.
\end{proof}

Now we can describe an efficient algorithm that, for any $\eps>0$, outputs an additive $\eps$-estimate of $\beta_k/d_k$ with high probability.
The correctness and complexity of the algorithm are proven in \Cref{th:algo}.

\begin{algorithm}[ht!]
    \caption{Algorithm for $\eps$-estimating $\beta_k/d_k$}
    \label{alg:cheb}
    \KwIn{Query and sample access to complex $K$, integer $k$, estimates of $\hat\lambda$ and $\gamma$, precision parameter $\eps \in (0,1)$.}
    \KwOut{Estimate $\tilde\nu_k$ such that $\tilde\nu_k = \beta_k/d_k \pm \eps$ with high probability.}
    \SetAlgoLined
    \SetKw{Continue}{continue}
    
    \vspace{6pt}
    
    Set $r = \left\lceil \frac{1}{\gamma} \log\frac{3}{\eps} \right\rceil$ and $d = \left\lceil \sqrt{\frac{2}{\gamma}} \log\frac{6}{\eps} \right\rceil$.
 
    \For{$\ell = 0,\dots,d$}{
    Estimate $\Tr(H^\ell)/d_k$ to additive precision $\delta = \eps/\left(3 (d+1) 2^{3d}\right)$ with high probability using Algorithm \ref{alg:simple}. Let $\text{est}_{k,\ell}$ denote the output.
    
  }
        
        Return $\tilde\nu_k = \sum_{\ell=0}^{d} b^{(r,d)}_\ell \text{est}_{k,\ell}$.\
    
\end{algorithm}

\begin{theo} \label{th:algo}
   Algorithm \ref{alg:cheb} returns with high probability an estimate of $\beta_k/d_k$ with additive error~$\eps$ in time~$n^{\Oh\left(\frac{1}{\sqrt{\gamma}} \log\frac{1}{\eps}\right)}$.
\end{theo}
\begin{proof}
First we prove the correctness.
By our choice of $r$, we know from \Cref{obs:tr-bk} that $\Tr(H^r)/d_k = \beta_k/d_k \pm \eps/3$, so it suffices to return an $(2\eps/3)$-additive estimate of $\Tr(H^r)/d_k$.
By \Cref{lem:cheb-approx} we can use the approximation
\[
\frac{1}{d_k} \Tr(H^r)
= \frac{1}{d_k} \sum_{\ell=0}^d b^{(r,d)}_\ell \Tr(H^\ell) \pm \eps/3
\]
for $d = \left\lceil \sqrt{2 r \log \frac{6}{\eps}} \right\rceil \leq \left\lceil \sqrt{\frac{2}{\gamma}} \log \frac{6}{\eps} \right\rceil$.
We estimate each term $\Tr(H^\ell)/d_k$ to precision $\delta = \frac{\eps}{3 (d+1) 2^{3d}}$ with high probability, so that the final estimator has a total error
\begin{align*}
    \tilde\nu_k
&= \sum_{\ell=0}^{d} b^{(r,d)}_\ell \text{est}_{k,\ell}
= \sum_{\ell=0}^{d} b^{(r,d)}_\ell \left( \Tr(H^\ell)/d_k \pm \delta \right) \\
&= \left( \frac{1}{d_k} \sum_{\ell=0}^{d} b^{(r,d)}_\ell \Tr(H^\ell) \right) \pm \eps/3,
\end{align*}
using that $\left| \sum_{\ell=0}^{d} b^{(r,d)}_\ell \delta \right| \le \delta (d+1) 2^{3d} \le \eps/3$.
Combined with the previous error bounds this shows that~$\tilde\nu_k = \beta_k/d_k \pm \eps$ with high probability.

To bound the time complexity, recall that the time complexity of Algorithm \ref{alg:simple} in Line 3 is~$\Oh(n^{4\ell}/\delta^2) \in n^{\Oh(d)}/\eps^2$.
Summing over the $d+1$ loops, and using the expression for $d$, this yields a total time complexity that is $n^{\Oh\left(\frac{1}{\sqrt{\gamma}} \log \frac{1}{\eps} \right)}$.
\end{proof}

This completes the proof of the first item of \Cref{thm:main}.

\section{Algorithm for clique complexes \label{sec:clique}}

The complexity of our path integral Monte Carlo algorithm is dominated by the sample complexity that follows from Hoeffding's inequality (\Cref{lem:hoeffdingineq}), which we bound using the fact that $|Y_z| \leq \|H\|_1^z$ and $\|H\|_1 = \poly(n)$.
Here we prove a tighter bound on $\|H\|_1$ for the special case of clique complexes, and exploit this to improve the algorithm.

We will use the following characterization of the off-diagonal elements of the combinatorial Laplacian $\Dk$.
\begin{lem}[Follows from \Cref{lem:laplacemat}] \label{lemma:nonzero_laplacian}
    Let $\Dk$ denote the $k$-th combinatorial Laplacian of a simplicial complex $K$.
    Then $(\Dk)_{ij}$ for $i\neq j$ is nonzero if and only if the corresponding two $k$-faces are down-up neighbors but not up-down neighbors.
\end{lem}

The following claim is going to be useful for the proof of the next lemma. 
\begin{cl} \label{claim:neighbors}
        In a clique complex, every $k$-face has at most $n-k-1$ down-up neighbors that are not its up-down neighbors.
    \end{cl}
    
    \begin{proof}
        Since we are in the clique complex case, a $k$-face is exactly a $(k+1)$-clique, so we will use the two expressions interchangeably. We will prove a slightly stronger statement: every vertex that is not in a $(k+1)$-clique $C$ can appear in at most one down-up neighbor of $C$ that is not its up-down neighbor. Let us consider a $(k+1)$-clique $C$ and suppose there is a vertex $v$ among the $n-k-1$ vertices that are not in $C$ such that $v$ belongs to two distinct down-up neighbors $C_1$ and $C_2$ of~$C$. That is, suppose there are two distinct vertices $u_1,u_2 \in C$ such that $C_1 = C \setminus \{u_1\} \cup \{v\}$ and $C_2 = C \setminus \{u_2\} \cup \{v\}$ are $(k+1)$-cliques. We show that $C_1$ and $C_2$ are up-down neighbors of $C$. Indeed, $v$ must be adjacent to every vertex of $C$: from $C_1 \in K$ it is adjacent to all vertices in $C$ other than the vertex $u_1$, and from $C_2 \in K$ it is adjacent all vertices except for $u_2$.
        So $C \cup \{v\}$ forms a $(k+2)$-clique and hence $C_1$ and $C_2$ are up-down neighbors of $C$. 
    \end{proof}

This section's main observation is the following.
\begin{lem} \label{lem_clique_nonzeros}
    The $k$-th combinatorial Laplacian of a clique complex has at most $n-k-d^{\text{up}}_i$ nonzero entries in every row, that is, 
    $$|\{j : (\Dk)_{ij}\neq0\}|\leq n-k-d^{\text{up}}_i \quad \forall i\in d_k$$
    where $d^{\text{up}}_i$ is the up-degree of the $k$-face corresponding to the $i$-th row of the combinatorial Laplacian.
\end{lem}

\begin{proof}
    As a consequence of \Cref{claim:neighbors}, every $k$-face has at most $n-k-1$ down-up neighbors that are not its up-down neighbors. Following \Cref{lemma:nonzero_laplacian}, these elements correspond exactly to the nonzero off-diagonal entries in $\Dk$. Adding the diagonal element $(\Dk)_{ii}$, we obtain that the total number of nonzero entries in a row of the $k$-th combinatorial Laplacian is at most $n-k$.
    
    To improve this bound, notice that if a vertex $v$ is adjacent to all the vertices of a $k$-face (i.e, we have an up-neighboring $(k+1)$\nobreakdash-face), then $v$ cannot be in any down-up neighbor that is not an up-down neighbor as well. Thus, using \Cref{lemma:nonzero_laplacian} again, we can say that every up-neighbor ``cancels'' the corresponding down-up neighbors in $\Dk$.
    
 Hence, if the $k$-face corresponding to the $i$-th row of the combinatorial Laplacian has up-degree~$d^{\text{up}}_i$, then the number of nonzero entries in the $i$-th row of $\Dk$ is at most $n-k-d^{\text{up}}_i$. 
\end{proof}

From this, we get the following corollary.

\begin{coro} \label{obs_Hleq2}
    $\|H\|_1 \leq 2n/\hat\lambda$.
\end{coro}
\begin{proof}
    Recall that $H= I - \Dk/\hat\lambda$. Thus, 
 \begin{align*}
     \|H\|_1 
    &= \max_j \sum_i |H_{ij}| \\
    &\leq \max_j \left|(I)_{jj} - \frac{(\Dk)_{jj}}{\hat\lambda}\right| + n\cdot \frac{1}{\hat\lambda}
    \leq 2 \frac{n}{\hat\lambda},
 \end{align*}
    where for the first inequality we used the fact that $|(\Dk)_{ij}|$ is either 1 or 0 if $i\neq j$, and by \Cref{lemma:nonzero_laplacian} it is 1 at most $n$ times in every row or column. In the second inequality we used the fact that $0\leq (\Dk)_{ii} \leq n$. Combining, we have the result.
\end{proof}

Now let us recall the path integral estimator $Y_z$ as defined in the previous section:
$$Y_z(j_0,j_1,\dots,j_z)
= \braket{j_0|j_z} \prod_{\ell=0}^{z-1} (-1)^{s(j_{\ell+1},j_\ell)} \|H_{\cdot,j_\ell}\|_1.$$
As a consequence of \Cref{obs_Hleq2}, it satisfies
\[
|Y_z|
\leq \|H\|_1^z
\leq \left(\frac{2n}{\hat\lambda}\right)^{z}.
\]
This improves the sample complexity in \Cref{lem:Hoeffding2} from $\Oh\left(n^{4z}\right)/\delta^2$ to $\Oh\left(\left(\frac{2n}{\hat\lambda}\right)^{2z}\right) \cdot \frac{1}{\delta^2}$ that is $\left(\frac{n}{\hat\lambda}\right)^{\Oh(z)} \cdot \frac{1}{\delta^2}$.
This directly propagates to Algorithm \ref{alg:cheb}, improving its time complexity from
\[
n^{\Oh\left(\frac{1}{\sqrt{\gamma}} \log\frac{1}{\eps}\right)}
\text{ to }
\left(\frac{n}{\hat\lambda}\right)^{\Oh\left(\frac{1}{\sqrt{\gamma}} \log\frac{1}{\eps}\right)}\cdot\poly(n).
\]
The $\poly(n)$ term comes from the time required for sampling (\Cref{lemma_sample}). This completes the proof of the second item of \Cref{thm:main}.

\section*{Acknowledgements}
This work benefited from discussions with Chris Cade, Marcos Crichigno, and Sevag Gharibian. Sayantan would like to thank Sophie Laplante for hosting him at IRIF, Paris, for an academic visit where this work was initiated.
Finally, we thank anonymous reviewers for pointing out useful references and ways to improve our presentation.

\bibliographystyle{plainnat}
\bibliography{bibliography}

\end{document}